\documentclass{ws-ijfcs}

\usepackage{amsmath,amssymb,amsfonts,url,algorithm,graphicx}
\usepackage{wasysym}
\usepackage[T1]{fontenc}
\usepackage[utf8]{inputenc}
\usepackage[noend]{algorithmic}

\allowdisplaybreaks

\DeclareMathOperator{\Ker}{Ker}
\DeclareMathOperator{\Pre}{Pre}

\DeclareMathOperator*{\argmin}{arg\,min}
\DeclareMathOperator{\merge}{merge}

\providecommand*{\abs}[1]{\ensuremath{\lvert #1 \rvert}}
\providecommand*{\seq}[3]{\ensuremath{#1_{#2}, \dotsc, #1_{#3}}}
\providecommand*{\word}[3]{\ensuremath{#1_{#2} \dotsm #1_{#3}}}
\providecommand{\simdiff}[0]{\ensuremath{\mathop{\triangle}}}

\providecommand*{\nat}[0]{\ensuremath{\mathbb N}}
\providecommand*{\integer}[0]{\ensuremath{\mathbb Z}}

\begin{document}

\markboth{A.\ Maletti and D.\ Quernheim}{Optimal Hyper-Minimization}

%
\catchline{}{}{}{}{}
%

\title{OPTIMAL HYPER-MINIMIZATION\footnote{This is an extended and
    revised version of [A.\ Maletti: \emph{Better hyper-minimization —
      not as fast, but fewer errors}. In Proc.\ CIAA, volume 6482 of
    LNCS, pages 201-210. Springer-Verlag, 2011].}} 
\author{ANDREAS MALETTI\footnote{The work was carried out while the
    author was at the \emph{Departament de Filologies Rom\`aniques,
      Universitat Rovira i Virgili} (Tarragona, Spain) and was
    supported by the \emph{Ministerio de Educaci\'on y Ciencia} (MEC)
    grants JDCI-2007-760 and MTM-2007-63422.}~ \and DANIEL QUERNHEIM} 
\address{Institute for Natural Language Processing, Universit\"at
  Stuttgart \\ Azenbergstra\ss e~12, 70174 Stuttgart, Germany \\
  \email{\{andreas.maletti, daniel.quernheim\}@ims.uni-stuttgart.de}}

\maketitle

\begin{history}
\received{(31 January 2011)}
\accepted{(Day Month Year)}
\comby{(xxxxxxxxxx)}
\end{history}

\begin{abstract}
  Minimal deterministic finite automata~(\textsc{dfa}s) can be reduced
  further at the expense of a finite number of errors.  Recently, such
  minimization algorithms have been improved to run in time~$O(n \log
  n)$, where $n$~is the number of states of the input \textsc{dfa}, by
  [\textsc{Gawrychowski} and \textsc{Je\.{z}}: Hyper-minimisation made
  efficient.  Proc.\ \textsc{Mfcs}, \textsc{Lncs}~5734, 2009] and
  [\textsc{Holzer} and \textsc{Maletti}: An {$n \log n$} algorithm for
  hyper-minimizing a (minimized) deterministic
  automaton. \emph{Theor.\ Comput.\ Sci.}~411, 2010].  Both algorithms
  return a \textsc{dfa} that is as small as possible, while only
  committing a finite number of errors.  These algorithms are further
  improved to return a \textsc{dfa} that commits the least number of
  errors at the expense of an increased (quadratic) run-time.  This
  solves an open problem of [\textsc{Badr}, \textsc{Geffert}, and
  \textsc{Shipman}: Hyper-minimizing minimized deterministic finite
  state automata. \emph{\textsc{Rairo} Theor.\ Inf.\ Appl.}~43,
  2009].  In addition, an experimental study on random automata is
  performed and the effects of the existing algorithms and the new
  algorithm are reported.
\end{abstract}

\keywords{deterministic finite automaton; minimization; error analysis.}

\ccode{2010 Mathematics Subject Classification: 68Q45, 68Q25, 68W40}

\section{Introduction}
\label{sec:Intro}
Deterministic finite automata (\textsc{dfa}s)~\cite{yu97} are used in
a vast number of applications that require huge automata like speech
processing~\cite{moh97} or linguistic analysis~\cite{joh72}.  To keep
the operations efficient, minimal \textsc{dfa} are typically used in
applications.  A minimal \textsc{dfa} is such that all equivalent
\textsc{dfa}s are larger, where the size is measured by the number of
states.  The asymptotically fastest minimization algorithm runs in
time~$O(n \log n)$ and is due to \textsc{Hopcroft}~\cite{hopull79},
where $n$~is the size of the input \textsc{dfa}.

Recently, stronger minimization procedures, called hyper-minimization,
have been investigated~\cite{badgefshi07,bad09,gawjez09,holmal10,sch10}.
They can efficiently compress minimal \textsc{dfa}s even further at
the expense of a finite number of errors.  The fastest
hyper-minimization algorithms~\cite{gawjez09,holmal10} run in
time~$O(n \log n)$.  More specifically, given an input
\textsc{dfa}~$M$, a hyper-minimization algorithm returns a
\emph{hyper-minimal \textsc{dfa} for~$M$}, which
\begin{itemize}
\item recognizes the same language as~$M$ up to a finite number of
  errors, and
\item is minimal among all \textsc{dfa}s with the former property
  (hyper-minimal).
\end{itemize}

In this contribution, we extend a known hyper-minimization algorithm
to return a \emph{hyper-optimal \textsc{dfa} for~$M$}, which is a
hyper-minimal \textsc{dfa} for~$M$ that commits the least number of
errors among all hyper-minimal \textsc{dfa}s for~$M$.  Moreover, the
algorithm returns the number of committed errors, which allows a user
to disregard the returned \textsc{dfa} if the number is unacceptably
large.  Our algorithm is based essentially on a syntactic
characterization of hyper-minimal \textsc{dfa}s for~$M$ (see Theorems
3.8~and~3.9 of~\cite{badgefshi07}).  Roughly speaking, two
hyper-minimal \textsc{dfa}s for~$M$ differ in exactly three
aspects~\cite{badgefshi07}: (i)~the finality of the states~$P$ that
are reachable by only finitely many strings, (ii)~the transitions from
states of~$P$ to states not in~$P$, and (iii)~the initial state.  The
characterization has two main uses: It allows us to compute the exact
number of errors for each hyper-minimal \textsc{dfa} for~$M$, and it
allows us to easily consider all hyper-minimal \textsc{dfa}s for~$M$
in order to find a hyper-optimal \textsc{dfa} for~$M$.  We thus solve
a remaining open problem of~\cite{badgefshi07}.  Unfortunately, the
time complexity of the obtained algorithm is~$O(n^2)$, and it remains
an open problem whether the algorithm can be improved to run in
time~$O(n \log n)$.

Finally, we demonstrate hyper-minimization and the new algorithm on
test \textsc{dfa}s, which we generated from random non-deterministic
finite automata~\cite{yu97,tabvar05}.  The difficult cases for
minimization that were identified in~\cite{tabvar05} also prove to be
difficult for hyper-minimization in the sense that only a small
reduction is possible at the expense of a significant amount of
errors.  The new algorithm alleviates this problem by avoiding a large
number of mistakes.  Outside the hard instances of~\cite{tabvar05},
already hyper-minimization reduces the size nicely at the expense of
only a few errors.

\section{Preliminaries}
\label{sec:Prelim}
The set of integers is~$\integer$, and the subset of nonnegative
integers is~$\nat$.  If the symmetric difference $S \simdiff T = (S
\setminus T) \cup (T \setminus S)$ of two sets $S$~and~$T$ is finite,
then $S$~and~$T$ are almost-equal.  Each finite set~$\Sigma$ is an
alphabet, and the set of all strings over~$\Sigma$ is~$\Sigma^*$.  The
empty string is $\varepsilon$, and the concatenation of two strings
$u, v \in \Sigma^*$ is denoted by the juxtaposition~$uv$.  The length
of the string~$w = \word \sigma 1k$ with $\seq \sigma1k \in \Sigma$
is~$\abs w = k$.  A string~$u \in \Sigma^*$ is a prefix of~$w$ if
there exists a string~$v \in \Sigma^*$ such that $w = uv$.  Any subset
$L \subseteq \Sigma^*$ is a language over~$\Sigma$.

A deterministic finite automaton (for short: \textsc{dfa}) is a
tuple~$M = (Q, \Sigma, q_0, \delta, F)$, in which $Q$~is a finite set
of states, $\Sigma$ is an alphabet of input symbols, $q_0 \in Q$ is an
initial state, $\delta \colon Q \times \Sigma \to Q$ is a transition
mapping, and $F \subseteq Q$ is a set of final states.  The transition
mapping~$\delta$ extends to a mapping $\underline\delta \colon Q
\times \Sigma^* \to Q$ by $\underline\delta(q, \varepsilon) = q$ and
$\underline\delta(q, \sigma w) = \underline\delta(\delta(q, \sigma),
w)$ for every $q \in Q$, $\sigma \in \Sigma$, and $w \in \Sigma^*$.
For every $q \in Q$, let
\[ L(M, q) = \{ w \in \Sigma^* \mid \underline\delta(q_0, w) = q \}
\qquad \text{and} \qquad L(q, M) = \{ w \in \Sigma^* \mid
\underline\delta(q, w) \in F\} \enspace. \] Intuitively, $L(M,
q)$~contains all strings that take~$M$ (from the initial state~$q_0$)
into the state~$q$, and $L(q, M)$~contains all strings that take~$M$
from~$q$ into a final state.  Moreover, $\Ker(M) = \{ q \in Q \mid
L(M, q) \text{ infinite} \}$ is the set of kernel states of~$M$, and
$\Pre(M) = Q \setminus \Ker(M)$ is the set of preamble states.  The
sets $\Ker(M)$~and~$\Pre(M)$ can be computed in time~$O(m)$, where $m
= \abs{Q \times \Sigma}$.  The \textsc{dfa}~$M$ recognizes the
language $L(M) = L(q_0, M) = \bigcup_{q \in F} L(M, q)$.

An equivalence relation~$\mathord{\equiv} \subseteq S \times S$ is a
reflexive, symmetric, and transitive binary relation.  The equivalence
class of an element~$s \in S$ is $[s]_\equiv = \{ s' \in S \mid s
\equiv s'\}$ and $[S]_\equiv = \{ [s]_\equiv \mid s \in S\}$.  A weak
partition of~$S$ is a set~$\Pi$ such that (i)~$A \subseteq S$ for
every $A \in \Pi$, (ii)~$A_1 \cap A_2 = \emptyset$ for all different
$A_1, A_2 \in \Pi$, and (iii)~$S = \bigcup_{A \in \Pi} A$.  An
equivalence relation $\mathord{\equiv} \subseteq Q \times Q$ on the
states of the \textsc{dfa}~$M = (Q, \Sigma, q_0, \delta, F)$ is a
congruence relation on~$M$ if $\delta(q_1, \sigma) \equiv \delta(q_2,
\sigma)$ for all $q_1 \equiv q_2$ and $\sigma \in \Sigma$.

Let $M = (Q, \Sigma, q_0, \delta, F)$ and $N = (P, \Sigma, p_0, \mu,
G)$ be two \textsc{dfa}s.  A mapping $h \colon Q \to P$ is a
transition homomorphism if $h(\delta(q, \sigma)) = \mu(h(q), \sigma)$
for every $q \in Q$ and $\sigma \in \Sigma$.  If additionally $q \in
F$ if and only if $h(q) \in G$ for every $q \in Q$, then $h$~is a
(\textsc{dfa}) homomorphism.  In both cases, $h$~is an isomorphism if
it is bijective.  Finally, we say that the \textsc{dfa}s $M$~and~$N$
are (transition and \textsc{dfa}) isomorphic if there exists a
(transition and \textsc{dfa}, respectively) isomorphism~$h \colon Q
\to P$.

The \textsc{dfa}s $M$~and~$N$ are equivalent if $L(M) = L(N)$.
Clearly, (\textsc{dfa}) isomorphic \textsc{dfa}s are equivalent.  Two
states $q \in Q$ and $p \in P$ are equivalent, denoted by~$q \equiv
p$, if $L(q, M) = L(p, N)$.\footnote{While it might not be clear from
  the notation $q \equiv p$ to which \textsc{dfa} a state belongs, it
  will typically be clear from the context.  In particular, we might
  have $M = N$; i.e., we might relate two states from the same
  \textsc{dfa}.}  The equivalence~$\mathord{\equiv} \subseteq Q \times
Q$ is a congruence relation on~$M$.  The \textsc{dfa}~$M$ is minimal
if it does not have equivalent states (i.e., $q_1 \equiv q_2$ implies
$q_1 = q_2$ for all $q_1, q_2 \in Q$).  The name `minimal' is
justified by the fact that there does not exist a \textsc{dfa} with
strictly fewer states that recognizes the same language as a minimal
\textsc{dfa}.  A minimal \textsc{dfa} that is equivalent to~$M$ can be
computed efficiently using \textsc{Hopcroft}'s algorithm~\cite{hop71},
which runs in time~$O(m \log n)$ where $m = \abs{Q \times \Sigma}$ and
$n = \abs Q$.  Moreover, minimal \textsc{dfa}s are equivalent if and
only if they are isomorphic.

Similarly, the \textsc{dfa}s $M$~and~$N$ are almost-equivalent if
$L(M)$~and~$L(N)$ are almost-equal.  The states $q \in Q$ and $p \in
P$ are almost-equivalent, which is denoted by $q \sim p$, if $L(q,
M)$~and~$L(p, M)$ are almost-equal.  The
almost-equivalence~$\mathord{\sim} \subseteq Q \times Q$ is also a
congruence.  The minimal \textsc{dfa}~$M$ is hyper-minimal if it does
not have a pair $(q_1, q_2) \in Q \times Q$ of different, but
almost-equivalent states such that $\{q_1, q_2\} \cap \Pre(M) \neq
\emptyset$.  Again, the name `hyper-minimal' is justified by the fact
that there does not exist a \textsc{dfa} with strictly fewer states
that recognizes an almost-equivalent language (see Theorem~3.4
of~\cite{badgefshi07}).  A hyper-minimal \textsc{dfa} that is
almost-equivalent to~$M$ is called ``hyper-minimal for~$M$'' and can
be computed efficiently using the algorithms
of~\cite{gawjez09,holmal10}, which also run in time~$O(m \log n)$.  A
structural characterization of hyper-minimal \textsc{dfa}s is
presented in Theorems 3.8~and~3.9 of~\cite{badgefshi07}, which we
reproduce here.

\begin{theorem}[{\protect{see~\cite{badgefshi07}}}]
  \label{thm:Struc}
  Let $M = (Q, \Sigma, q_0, \delta, F)$ and $N = (P, \Sigma, p_0, \mu,
  G)$ be almost-equivalent \textsc{dfa}s.  Then $\underline\delta(q_0,
  w) \sim \underline\mu(p_0, w)$ for every $w \in \Sigma^*$.  In
  addition, if $M$~and~$N$ are hyper-minimal, then there exists a
  mapping $h \colon Q \to P$ such that 
  \begin{itemize}
  \item $q \sim h(q)$ for every $q \in Q$,
  \item $h$ yields a transition isomorphism between
    $\Pre(M)$~and~$\Pre(N)$, and
  \item $h$ yields a \textsc{dfa} isomorphism between
    $\Ker(M)$~and~$\Ker(N)$.
  \end{itemize}
\end{theorem}

\begin{algorithm}[t]
  \begin{algorithmic}[2]
    \REQUIRE a \textsc{dfa}~$M = (Q, \Sigma, q_0, \delta, F)$ with $m
    = \abs{Q \times \Sigma}$ and $n = \abs Q$ \smallskip 
    \STATE $M \gets \textsc{Minimize}(M)$
      \COMMENT{\textsc{Hopcroft}'s algorithm; $O(m \log n)$}
    \STATE $\mathord\sim \gets \textsc{CompAEquiv}(M)$
      \COMMENT{compute almost-equivalence; $O(m \log n)$}
    \STATE $M \gets \textsc{MergeStates}(M, \Ker(M), \mathord\sim)$
      \COMMENT{merge almost-equivalent states; $O(m)$}
    \RETURN $M$
  \end{algorithmic}
  \caption{Structure of a hyper-minimization algorithm.} 
  \label{alg:Overall}
\end{algorithm}

\section{Hyper-minimization}
\label{sec:Hyper}
Hyper-minimization as introduced in~\cite{badgefshi07} is a form of
lossy compression with the goal of reducing the size of a minimal
\textsc{dfa} at the expense of a finite number of errors.  More
formally, hyper-minimization aims to find a hyper-minimal \textsc{dfa}
for an input \textsc{dfa}.  Several hyper-minimization algorithms
exist~\cite{badgefshi07,bad09,gawjez09,holmal10}, and the
overall structure of the hyper-minimization algorithm
of~\cite{holmal10} is displayed in Algorithm~\ref{alg:Overall}.  For
the following discussion let $M = (Q, \Sigma, q_0, \delta, F)$ be a
\textsc{dfa}, and let $m = \abs{Q \times \Sigma}$ and $n = \abs Q$ be
the number of its transitions and the number of its states,
respectively.

The most interesting component of Algorithm~\ref{alg:Overall} is the
merging process. In general, the merge of a state $p \in Q$ into
another state $q \in Q$ redirects all incoming transitions of~$p$
to~$q$.  If $p = q_0$ then $q$~is the new initial state.  The finality
of~$q$ is not changed even if $p$ is final.  Clearly, the state~$p$
can be deleted after the merge if $p \neq q$.  Formally, $\merge_M(p
\to q) = (P, \Sigma, p_0, \mu, F)$, where $P = (Q \setminus \{p\})
\cup \{q\}$ and for every $q' \in Q$ and $\sigma \in \Sigma$
\[ p_0 =
\begin{cases}
  q & \text{if } q_0 = p \\
  q_0 & \text{otherwise}
\end{cases}
\qquad \text{and} \qquad
\mu(q', \sigma) =
\begin{cases}
  q & \text{if } \delta(q', \sigma) = p \\
  \delta(q', \sigma) & \text{otherwise.}
\end{cases}
\]

\begin{lemma}
  \label{lm:MergeErrors}
  Let $p, q \in Q$ and $N = \merge_M(p \to q)$.  Then 
  \[ L(M) \simdiff L(N) = \{ uw \mid u \in L(M, p), w \in
  L(p, M) \simdiff L(q, M) \} \enspace. \]
\end{lemma}

Consequently, $M$~and~$\merge_M(p \to q)$ are almost-equivalent if $q
\sim p$ and $p \in \Pre(M)$.  The hyper-minimization algorithms
of~\cite{badgefshi07,bad09,gawjez09,holmal10} only perform such
merges.  More precisely, the procedure \textsc{MergeStates} merges
almost-equivalent states in the mentioned fashion until the obtained
\textsc{dfa} is hyper-minimal.  The number of errors introduced in
this way differs among several hyper-minimal \textsc{dfa} for~$M$ and
depends on the merges performed.  In this contribution, we develop an
algorithm that computes a hyper-minimal \textsc{dfa} for~$M$ that
commits the minimal number of errors among all hyper-minimal
\textsc{dfa}s for~$M$.  A \textsc{dfa}~$N$ is \emph{hyper-optimal
  for~$M$} if it is hyper-minimal and the cardinality of the symmetric
difference between $L(M)$~and~$L(N)$ is minimal among all
hyper-minimal \textsc{dfa}s.  Note that a hyper-optimal \textsc{dfa}
for~$M$ is hyper-minimal for~$M$.  Moreover, our algorithm returns the
exact number of errors, and we could also return a compact
representation of the actual error strings.  Overall, we thus solve a
problem that remained open in~\cite{badgefshi07}.

\begin{figure}[t]
  \centering
  \includegraphics[scale=1]{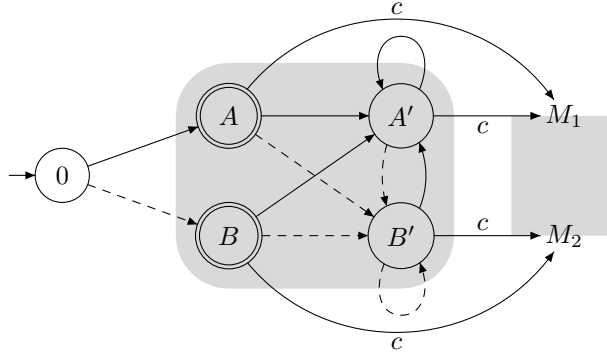}
  \caption{An example \textsc{dfa}, where unbroken lines are
    $a$-transitions and dashed lines are $b$-transitions.} 
  \label{fig:Original}
\end{figure}

An extreme example is presented in Fig.~\ref{fig:Original}.  If we
run the hyper-minimization algorithms
of~\cite{badgefshi07,bad09,gawjez09,holmal10}, then we obtain
one of the two first \textsc{dfa}s of Fig.~\ref{fig:Mini}.  Both of
them commit $2 + \abs{L(M_1) \simdiff L(M_2)}$~errors.  If we let
$L(M_1) = \Sigma^k$ for some $k \in \nat$ and $L(M_2) = \emptyset$,
then they commit $2 + \abs \Sigma^k$~errors.  On the other hand, the
optimal \textsc{dfa} is the third \textsc{dfa} of
Fig.~\ref{fig:Mini}, and it commits only $2$~errors (irrespective of
$M_1$~and~$M_2$).  This shows that the gap in the number of errors can
be very significant.

\begin{figure}[t]
  \centering
  \includegraphics[scale=0.85]{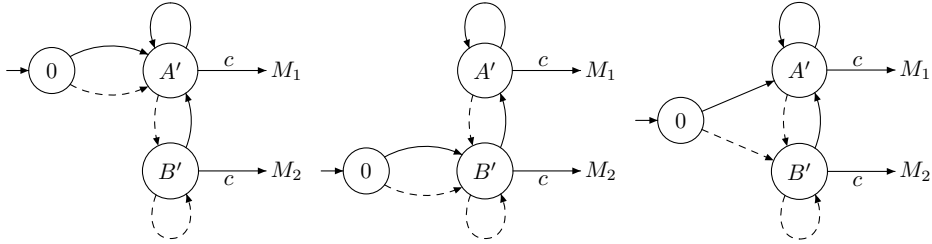} 
  \caption{Three hyper-minimal \textsc{dfa}s for the \textsc{dfa} of
    \protect{Fig.~\ref{fig:Original}}, where unbroken lines are
    $a$-transitions and dashed lines are $b$-transitions.}
  \label{fig:Mini}
\end{figure}

\section{Computing the number of errors}
\label{sec:Errors}
Next, we show how to efficiently compute the number of errors that are
caused by a single merge (see Lemma~\ref{lm:MergeErrors}).  For this
we first compute the size of the difference between almost-equivalent
states $p \sim q$.  From now on, let $M = (Q, \Sigma, q_0, \delta, F)$
be a minimal \textsc{dfa}.  In our examples, we will always refer to
our running example \textsc{dfa}~$M_{\text{ex}}$, which is presented
in Fig.~\ref{fig:Ex}.  Its kernel states are $\Ker(M_{\text{ex}}) =
\{ E, F, I, J, K, L, M\}$ and the following partition represents its
almost-equivalence:
\[ \{ 0 \} \quad \{ A \} \quad \{ B \} \quad \{ C, D \} \quad \{ E \}
\quad \{ F \} \quad \{ G, H, I, J \} \quad \{ K, L, M \}. \] In
comparison to the \textsc{dfa}~$M_{\text{ex}}$ of Fig.~\ref{fig:Ex},
the \textsc{dfa}~$N_{\text{ex}}$ of Fig.~\ref{fig:Ex} commits the
following seven errors: $\{ aaaab, aaab, aab, aabab, aabb, abab, abb \}$.
Note that existing
algorithms will only find hyper-minimal \textsc{dfa}s that commit
$16$~errors, and the worst hyper-minimal \textsc{dfa}
for~$M_{\text{ex}}$ commits $29$~errors.

\begin{figure}[t]
  \centering
  \includegraphics[scale=0.85]{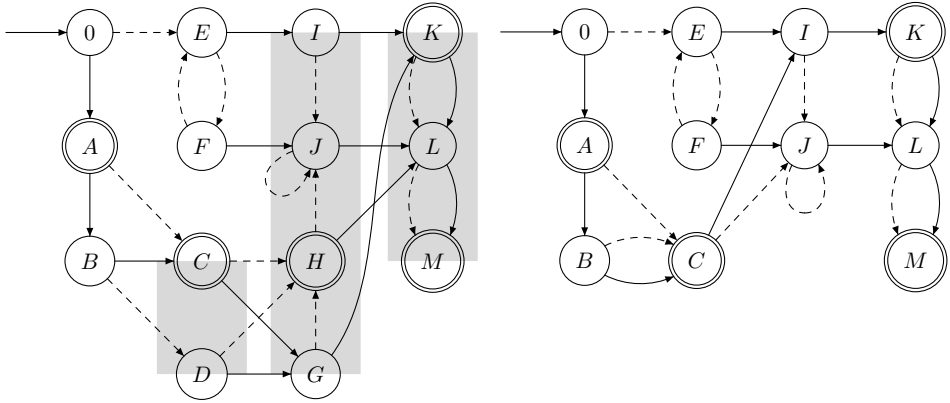} 
  \caption{Example \textsc{dfa}~$M_{\text{ex}}$ (left) and optimal
    hyper-minimal \textsc{dfa}~$N_{\text{ex}}$ (right)
    for~$M_{\text{ex}}$, where unbroken lines are $a$-transitions and
    dashed lines are $b$-transitions.}
  \label{fig:Ex}
\end{figure}

\begin{definition}
  \label{df:EMatrix}
  For every $q \sim p$, let
  \[ E_{q, p} =
  \begin{cases}
    0 & \text{if } q = p \\
    \sum_{\sigma \in \Sigma} E_{\delta(q, \sigma), \delta(p, \sigma)}
    +
    \begin{cases}
      0 & \text{if } q \in F \iff p \in F \\
      1 & \text{otherwise}
    \end{cases}
    & \text{otherwise.}
  \end{cases}
  \]
\end{definition}

\begin{lemma}
  \label{lm:EMatrix}
  $E_{q, p} = \abs{L(q, M) \simdiff L(p, M)}$ for every $q \sim p$.
\end{lemma}

\begin{proof}
  Let $q \sim p$.  Then $\abs{L(q, M) \simdiff L(p, M)}$ is finite by
  definition, and we let $k_{q,p} = \max\ \{ \abs w \mid w \in L(q, M)
  \simdiff L(p, M) \}$, where $\max\ \emptyset = -\infty$.  Now, we
  prove the statement by induction on~$\nat \cup \{-\infty\}$.  First,
  suppose that $k_{q, p} = -\infty$.  Then $L(q, M) = L(p, M)$, which
  yields that $q \equiv p$.  Since $M$~is minimal, we conclude that $q
  = p$ and $E_{q, p} = 0$, which proves the induction base.  Second,
  suppose that $k_{q,p} \geq 0$, and let $W = \{ \sigma w \mid \sigma
  \in \Sigma, w \in L(\delta(q, \sigma), M) \simdiff L(\delta(p,
  \sigma), M) \}$. Obviously, $W \subseteq L(q, M) \simdiff L(p, M)
  \subseteq W \cup \{\varepsilon\}$ and $k_{\delta(q, \sigma),
    \delta(p, \sigma)} < k_{q,p}$ for every $\sigma \in \Sigma$.  The
  empty string~$\varepsilon$ is in $L(q, M) \simdiff L(p, M)$ if and
  only if $q$~and~$p$ differ on finality.  Moreover, $E_{\delta(q,
    \sigma), \delta(p, \sigma)} = \abs{L(\delta(q, \sigma), M)
    \simdiff L(\delta(p, \sigma), M)}$ for every $\sigma \in \Sigma$
  by induction hypothesis.  Since $k_{q, p} \geq 0$, we have $q \neq
  p$ and
  \[ E_{q, p} = \sum_{\sigma \in \Sigma} E_{\delta(q, \sigma),
    \delta(p, \sigma)} +
  \begin{cases}
    0 & \text{if } q \in F \iff p \in F \\
    1 & \text{otherwise,}
  \end{cases}
  \]
  which proves the induction step and the statement.
\end{proof}

\begin{algorithm}[t]
  \begin{algorithmic}[2]
    \REQUIRE minimal \textsc{dfa} $M = (Q, \Sigma, q_0, \delta, F)$
    and states~$q \sim p$
    \ENSURE error matrix $E \in \integer^{Q \times Q}$ initially $0$
    on the diagonal and $-1$ elsewhere \smallskip
    \IF{$E_{q, p} = -1$}
      \STATE $c \gets ((q \in F) \text{ xor } (p \in F))$
        \COMMENT{set errors to~$1$ if $q$~and~$p$ differ on finality}
        \smallskip
      \STATE $\displaystyle E_{q, p} \gets c + \sum_{\sigma \in \Sigma}
        \textsc{CompE}(M, \delta(q, \sigma), \delta(p, \sigma))$
        \COMMENT{recursive calls}
    \ENDIF
    \RETURN $E_{q, p}$
      \COMMENT{return the computed value}
  \end{algorithmic}
  \caption{\textsc{CompE}: Compute the error matrix~$E$.}
  \label{alg:Errors}
\end{algorithm}

Let us illustrate Algorithm~\ref{alg:Errors} on the example
\textsc{dfa}~$M_{\text{ex}}$ of Fig.~\ref{fig:Ex}.  We list some error
matrix entries together with the corresponding error strings.  Note
that the error strings are not computed by the algorithm, but are
presented for illustrative purposes only.
\begin{alignat*}{4}
  E_{G, H} &= 5 \quad \{\varepsilon, a, aa, ab, b \} &\qquad
  E_{H, I} &= 4 \quad \{\varepsilon, a, aa, ab \} &\qquad
  E_{K, L} &= 3 \quad \{\varepsilon, a, b \}  \\  
  E_{G, I} &= 1 \quad \{b \} &\qquad
  E_{H, J} &= 1 \quad \{\varepsilon \} &\qquad
  E_{K, M} &= 2 \quad \{a, b \} \\
  E_{G, J} &= 4 \quad \{a, aa, ab, b \}  &\qquad
  E_{I, J} &= 3 \quad \{a, aa, ab \} &\qquad
  E_{L, M} &= 1 \quad \{\varepsilon \} &\qquad
\end{alignat*}
  
\begin{theorem}
  \label{thm:Errors}
  Algorithm~\ref{alg:Errors} can be used to compute all~$E_{q, p}$
  with $q \sim p$ in time~$O(mn)$. 
\end{theorem}

\begin{proof}
  Clearly, the initialization and the recursion for~$E_{q, p}$ are
  straightforward implementations of its definition
  (see~Definition~\ref{df:EMatrix}).  Moreover, each individual call
  takes only time~$O(\abs \Sigma)$ besides the time taken for the
  recursive calls.  Since each call computes one entry in the matrix
  and no entry is ever recomputed, we obtain the time complexity
  $O(\abs \Sigma \cdot n^2) = O(mn)$ because $m = \abs \Sigma \cdot
  n$.
\end{proof}

In addition, we need to compute the number of strings that lead to a
preamble state (see Lemma~\ref{lm:MergeErrors}).  This can
easily be achieved with a folklore algorithm (see
Algorithm~\ref{alg:NoPaths} and Lemma~4 of~\cite{epp95}) that computes
the number of paths from~$q_0$ to each preamble state.  Mind that the
graph of the \textsc{dfa}~$M$ restricted to its preamble
states~$\Pre(M)$ is acyclic.  Overall, the algorithm is very similar
to Algorithm~\ref{alg:Errors}, but we will not present a formal
comparison here.

\begin{algorithm}[t]
  \begin{algorithmic}[2]
    \REQUIRE a minimal \textsc{dfa}~$M = (Q, \Sigma, q_0, \delta, F)$ and
    a preamble state~$q \in \Pre(M)$
    \ENSURE access path vector $w \in \nat^Q$ initially $1$ at~$q_0$
    and $0$ elsewhere \smallskip
    \IF{$w_q = 0$}
      \STATE $w_q \gets \displaystyle\sum_{(p, \sigma) \in
        \delta^{-1}(q)} \textsc{CompAccess}(M, p)$
        \COMMENT{recursive calls}
    \ENDIF
    \RETURN $w_q$
      \COMMENT{return the computed value}
  \end{algorithmic}
  \caption{\protect{\textsc{CompAccess:} Compute the number of paths
      to a preamble state.}}
  \label{alg:NoPaths}
\end{algorithm}

\begin{theorem}[see~{\protect{\cite{epp95}}}]
  \label{thm:NoPaths}
  Algorithm~\ref{alg:NoPaths} can be used to compute the number of
  paths to each preamble state in time~$O(m)$.
\end{theorem}

\begin{proof}
  The correctness is obvious using the observation that $p$~is a
  preamble state for every $(p, \sigma) \in \delta^{-1}(q)$ with $q
  \in \Pre(M)$.  Clearly, the call~$\textsc{CompAccess}(M, q)$
  terminates in constant time if the value~$w_q$ has already been
  computed.  Moreover, each transition can be considered at most once
  in the sum in line~2, which yields the time complexity~$O(m)$. 
\end{proof}

Algorithm~\ref{alg:NoPaths} computes the following values for the
\textsc{dfa}~$M_{\text{ex}}$ of Fig.~\ref{fig:Ex}:
\begin{alignat*}{4}
  w_0 = w_A = w_B = w_D &= 1 \quad & w_C &= 2 \quad &
  w_G &= 3 \quad & w_H &= 6 \enspace.
\end{alignat*}

Overall, we can now efficiently compute the number of errors (or a
representation of the errors itself) caused by a single merge
operation.  However, multiple merges may affect each other.  An error
that is introduced by one merge might be removed by a subsequent
merge, so that we cannot simply obtain the exact error count by adding
the error counts for all performed merges.

\section{Optimal state merging}
\label{sec:Merge}
The previous section suggests how to compute a hyper-optimal
\textsc{dfa} for a given minimal \textsc{dfa}~$M = (Q, \Sigma, q_0,
\delta, F)$ with $m = \abs{Q \times \Sigma}$ and $n = \abs Q$.  We can
simply compute the exact set of errors for each hyper-minimal
\textsc{dfa} for~$M$ and select a \textsc{dfa} with a minimal error
count.  By Theorem~\ref{thm:Struc} we can easily enumerate all
hyper-minimal \textsc{dfa}s for~$M$, so that the above procedure would
be effective.  However, in this section, we show that we can also
obtain a hyper-optimal \textsc{dfa} using only local decisions.  This
is possible since the structural differences among hyper-minimal
\textsc{dfa}s for~$M$ mentioned in Theorem~\ref{thm:Struc} cause
different errors.  Roughly speaking, Theorem~\ref{thm:Struc} shows
that two hyper-minimal \textsc{dfa}s for~$M$ can only differ on
\begin{itemize}
\item the initial state,
\item finality of preamble states, and
\item transitions from preamble to kernel states.
\end{itemize}

Now, let us identify the strings and potential errors associated with
each of the three differences.  Recall that $\sim$~is the
almost-equivalence relating the states of~$M$.  To simplify the
following discussion, we introduce some additional notation.  For
every $q \in Q$, let $K_q = \{ p \in \Ker(M) \mid p \sim q\}$.  In
other words, the set~$K_q$ contains all kernel states that are
almost-equivalent to the state~$q$.  Moreover, let $P_\sim = \{ B \in
[Q]_\sim \mid B \subseteq \Pre(M) \}$ be the set of blocks of
almost-equivalent and exclusively preamble states.  Now we define sets
of strings that correspond to the three types of differences mentioned
above:
\begin{itemize}
\item Let $W_0 = \bigcup_{q \in K_{q_0}} \Sigma^*$.
\item Let $W_B = \bigcup_{q \in B} L(M, q)$ for every $B \in P_\sim$.
\item For every $B \in P_\sim$ and $\sigma \in \Sigma$ with
  $\bigcup_{q \in B} K_{\delta(q, \sigma)} \neq \emptyset$, let
  \[ W_{B, \sigma} = \{ u\sigma w \mid u \in W_B, w \in \Sigma^* \}
  \enspace. \]
\end{itemize}

\begin{lemma}
  \label{lm:Part}
  The following is a weak partition of~$\Sigma^*$:
  \begin{align*}
    \{W_0\} &\cup \{ W_B \mid B \in P_\sim \} \cup \{ W_{B,\sigma}
    \mid B \in P_\sim, \sigma \in \Sigma, \bigcup_{q \in B}
    K_{\delta(q, \sigma)} \neq \emptyset \} \enspace.
  \end{align*}
\end{lemma}

\begin{proof}
  Clearly, $W_0 = \Sigma^*$ if $K_{q_0} \neq \emptyset$ or $W_0 =
  \emptyset$ otherwise. Suppose the former; i.e., there exists $q \in
  K_{q_0}$.  Let $p \in \Pre(M)$ be a preamble state.  Since $M$~is
  minimal, there exists a string $w \in L(M, p)$.  Moreover, $p =
  \underline\delta(q_0, w) \sim \underline\delta(q, w)$ because $q_0
  \sim q$ and $\sim$~is a congruence.  Clearly, $\underline\delta(q,
  w)$ is a kernel state due to the fact that $q$~is a kernel state.
  Consequently, every preamble state $p \in \Pre(M)$ is
  almost-equivalent to some kernel state, which proves that $[p]_\sim
  \notin P_\sim$ for every $p \in \Pre(M)$.  This yields that
  the statement is correct if $K_{q_0} \neq \emptyset$.  

  In the second case, let $K_{q_0} = \emptyset$.  Then $W_0 =
  \emptyset$.  Clearly, $W_{B_1} \cap W_{B_2} = \emptyset$ for all
  different $B_1, B_2 \in P_\sim$ because $\{ L(M, q) \mid q \in Q\}$
  is a partition of~$\Sigma^*$.  Using the same reasoning, we can show
  that $W_{B_1}$~and~$W_{B_2, \sigma}$ are disjoint for all $B_1, B_2
  \in P_\sim$ and suitable $\sigma \in \Sigma$ using the additional
  observation that $K_{\underline\delta(q_0, w)} \neq \emptyset$ for
  every $w \in W_{B_2, \sigma}$, whereas $K_{\underline\delta(q_0, w)}
  = \emptyset$ for every $w \in W_{B_1}$.  Finally, let $B_1, B_2 \in
  P_\sim$ and suitable $\sigma_1, \sigma_2 \in \Sigma$.  Suppose that
  there exists $w \in W_{B_1, \sigma_1} \cap W_{B_2, \sigma_2}$.  When
  processing~$w$ by~$M$ there can only be one transition from a
  preamble state to a kernel state, which in both cases has to be
  achieved by the letter $\sigma_1 = \sigma_2$.  Moreover, the state
  before taking this transition is unique, which yields that also $B_1
  = B_2$.  Consequently, we have shown that all sets are disjoint.

  It remains to prove that all of~$\Sigma^*$ is covered.  Let $w \in
  \Sigma^*$ be an arbitrary string.  If $K_{\underline\delta(q_0, w)}
  = \emptyset$, then $w \in W_{[\underline\delta(q_0, w)]_\sim}$.  On
  the other hand, let $K_{\underline\delta(q_0, w)} \neq \emptyset$.
  Then there exists a prefix~$u$ of~$w$ such that
  $K_{\underline\delta(q_0, u)} \neq \emptyset$ and
  $K_{\underline\delta(q_0, v)} = \emptyset$ for all strict
  prefixes~$v$ of~$u$.  Then $w \in W_0$ if $u = \varepsilon$ and $w
  \in W_{[\underline\delta(q_0, v)]_\sim, \sigma}$ where $u = v\sigma$
  and $\sigma \in \Sigma$.  This concludes the proof.
\end{proof}

The previous lemma shows that error strings in the mentioned
sets are independent and cover all potential errors.  For our example
\textsc{dfa}~$M_{\text{ex}}$ of Fig.~\ref{fig:Ex} we have
\[ W_0 = \emptyset \qquad W_{\{C, D\}} = \{ aaa, aab, ab \} \qquad
W_{\{C, D\}, a} = \{ uaw \mid u \in W_{\{C, D\}}, w \in \Sigma^*\}
\enspace. \] Next we address all individual differences between
hyper-minimal \textsc{dfa}s for~$M$.  We start with the initial state.

\begin{lemma}
  \label{lm:Initial}
  If $K_{q_0} \neq \emptyset$, then each hyper-minimal \textsc{dfa}
  for~$M$ is obtained by pruning $\merge_M(q_0 \to q)$ for some $q \in
  K_{q_0}$.  Moreover, it commits exactly $E_{q_0, q}$~errors.
\end{lemma}

\begin{proof}
  Let $N = (P, \Sigma, p_0, \mu, G)$~be a hyper-minimal \textsc{dfa}
  for~$M$.  By Theorem~\ref{thm:Struc}, the \textsc{dfa}~$N$ consists
  of only kernel states and is isomorphic to the subautomaton of~$M$
  that is determined by~$\Ker(M)$.  Moreover, $q_0 \sim p_0$, which
  yields that $N$~is isomorphic to~$\merge_M(q_0 \to q)$ for some $q
  \in K_{q_0}$.  By Lemma~\ref{lm:MergeErrors} we have that $L(q, M) =
  L(q, N) = L(N)$ and $L(M) = L(q_0, M)$.  This yields that $L(M)
  \simdiff L(N) = L(q_0, M) \simdiff L(q, M)$, of which the size
  is~$E_{q_0, q}$ by Lemma~\ref{lm:EMatrix}.
\end{proof}

We can compute the number~$E_{q_0, q}$ of errors caused by the merge
of~$q_0$ into an almost-equivalent kernel state~$q \in K_{q_0}$ using
Algorithm~\ref{alg:Errors} of Section~\ref{sec:Errors}.  This simple
test is implemented in lines~1--2 of Algorithm~\ref{alg:OptMerge}.

Second, let us consider a block~$B \in P_\sim$ of almost-equivalent
preamble states.  Such a block must eventually be merged into a single
preamble state~$p$ in the hyper-minimal \textsc{dfa}~$N$, for which we
need to determine finality because the preamble states of two
hyper-minimal \textsc{dfa}s for~$M$ are only related by a
transition isomorphism (see Theorem~\ref{thm:Struc}). 

\begin{lemma}
  \label{lm:Finality}
  Let $B \in P_\sim$ and $N = (P, \Sigma, p_0, \mu, G)$ be a
  hyper-minimal \textsc{dfa} for~$M$.  Then $N$~commits either
  $\sum_{q \in B \cap F} w_q$~or~$\sum_{q \in B \setminus F} w_q$
  errors of~$W_B$.
\end{lemma}

\begin{proof}
  The set~$W_B$ contains all strings that take the \textsc{dfa}~$M$
  into some state of~$B$.  Moreover, all those strings take the
  hyper-minimal \textsc{dfa}~$N$ into a single state~$p \in P$; i.e.,
  $L(N, p) = W_B$ by Theorem~\ref{thm:Struc}.  Let
  \[ W'_B = \{ w \in W_B \mid w \in L(M) \} \qquad \text{and} \qquad
  W''_B = \{w \in W_B \mid w \notin L(M) \} \enspace; \] i.e., the
  partition into accepted and rejected strings (by~$M$) of~$W_B$,
  respectively.  Consequently, it is sufficient to compare the size of
  those sets because if $p \in G$ (i.e., $p$~is a final state of~$N$),
  then all strings of~$W''_B$ are errors.  This is due to the fact
  that they are rejected by~$M$, but accepted by~$N$.  On the other
  hand, the strings of~$W'_B$ are errors if $p$~is non-final.  Finally
  \begin{alignat*}{3}
    \abs{W'_B} &= \abs{ \{ w \in W_B \mid q \in F, w \in L(M, q) \}}
    \\
    &= \abs{ \{ w \in \Sigma^* \mid q \in B \cap F, w \in L(M, q) \}}
    &&= \sum_{q \in B \cap F} w_q \enspace,
  \end{alignat*}
  and similarly, $\abs{W''_B} = \sum_{q \in B \setminus F} w_q$. 
\end{proof}

Consequently, if and only if more strings are accepting (i.e.,
$\abs{W'_B} > \abs{W''_B}$), then the preamble state~$p \in P$ of~$N$
should be accepting.  This decision is codified in
Algorithm~\ref{alg:Finality}.  On our example
\textsc{dfa}~$M_{\text{ex}}$ of Fig.~\ref{fig:Ex} and the block~$B =
\{C, D\}$ it compares $W'_B = \{ aaa, ab \}$ and $W''_B = \{ aab \}$,
and thus decides that the state~$C$ of the
\textsc{dfa}~$N_{\text{ex}}$ of Fig.~\ref{fig:Ex} should be final.
Note that Lemma~\ref{lm:Part} shows that the errors are distinct for
different blocks $B_1$~and~$B_2$.  All of the following algorithms
will use the global variable~$e$, which will keep track of the number
of errors. Initially, it will be set to~$0$ and each discovered error
will increase it.  Finally, we assume that the vector~$w \in \nat^Q$
(see Algorithm~\ref{alg:NoPaths}) and the error matrix $E \in
\integer^{Q \times Q}$ (see Algorithm~\ref{alg:Errors}) have already
been computed and can be accessed in constant time.

\begin{algorithm}[t]
  \begin{algorithmic}[2]
    \REQUIRE a minimal \textsc{dfa}~$M = (Q, \Sigma, q_0, \delta, F)$ and
    a block~$B \in P_\sim$
    \ENSURE error count~$e$
    \smallskip
    \STATE $\displaystyle (\overline f, f) \gets \Bigl( \sum_{q \in B
      \cap F} w_q, \sum_{q \in B \setminus F} w_q \Bigr)$
      \COMMENT{errors for non-final and final state}
    \STATE $e \gets e + \min(\overline f, f)$
        \COMMENT{add smaller value to global error count}
    \STATE select $q \in B$ such that $q \in F$ if $\overline f > f$
        \COMMENT{select appropriate state}
    \RETURN $q$
      \COMMENT{return selected state}
  \end{algorithmic}
  \caption{\protect{\textsc{CompFinality}: Determine finality of a
      block of preamble states.}}
  \label{alg:Finality}
\end{algorithm}

\begin{lemma}
  \label{lm:Finality2}
  $\textsc{ComputeFinality}(M, B, w)$ adds the smallest number of
  errors of~$W_B$ committed by a hyper-minimal \textsc{dfa}~$N$
  for~$M$.  It runs in time~$O(\abs B)$ and returns a final state
  (of~$M$) if and only if $W_B \subseteq L(N)$.
\end{lemma}

\begin{proof}
  Algorithm~\ref{alg:Finality} implements the method of
  Lemma~\ref{lm:Finality} in the given run-time.
\end{proof}

For the third criterion, let us again consider a block~$B \in P_\sim$
of almost-equivalent preamble states and a symbol~$\sigma \in \Sigma$
such that $\bigcup_{q \in B} K_{\delta(q, \sigma)} \neq \emptyset$.
Clearly, $K_{\delta(q_1, \sigma)} = K_{\delta(q_2, \sigma)}$ for all
$q_1, q_2 \in B$ because $\sim$~is a congruence on~$M$.  We need to
determine the kernel state that will be the new transition target.  By
Theorem~\ref{thm:Struc} it has to be a kernel state because $\delta(q,
\sigma)$ is almost-equivalent to a kernel state.

\begin{lemma}
  \label{lm:Trans}
  Let $N = (P, \Sigma, p_0, \mu, G)$ be a hyper-minimal \textsc{dfa}
  for~$M$, and let $B \in P_\sim$ and $\sigma \in \Sigma$ be such that
  $K = \bigcup_{q \in B} K_{\delta(q, \sigma)} \neq \emptyset$.  Then
  the \textsc{dfa}~$N$ commits $\sum_{q \in B} w_q \cdot E_{\delta(q,
    \sigma), q'}$ errors of~$W_{B, \sigma}$ for some $q' \in K$.
\end{lemma}

\begin{proof}
  Since $W_{B, \sigma} = \{ u\sigma v \mid u \in W_B, v \in
  \Sigma^*\}$, each string~$w \in W_{B, \sigma}$ has a
  prefix~$u\sigma$ with $u \in W_B$.  Clearly, each $u \in W_B$ takes
  the \textsc{dfa}~$M$ into some state of~$B$, and the hyper-minimal
  \textsc{dfa}~$N$ into a state state~$p \in P$ such that $L(N, p) =
  W_B$ by Theorem~\ref{thm:Struc}.  Moreover, $\mu(p, \sigma) = p'$
  for some $p' \in \Ker(N)$, for which an equivalent state $q' \in Q$
  exists in~$M$ by Theorem~\ref{thm:Struc} because the kernels of
  $M$~and~$N$ are \textsc{dfa} isomorphic.  Consequently, $L(p', N) =
  L(q', M)$ and $N$~accepts the strings
  \[ \{ u \sigma v \mid u \in W_B, v \in L(q', M) \} \subseteq W_{B,
    \sigma} \] and rejects the remaining strings of~$W_{B, \sigma}$.
  On the other hand, the \textsc{dfa}~$M$ accepts the strings
  $\bigcup_{q \in B} \{ u \sigma v \mid u \in L(M, q), v \in
  L(\delta(q, \sigma), M) \} \subseteq W_{B, \sigma}$ and rejects the
  remaining strings of~$W_{B, \sigma}$.  Clearly, $\delta(q, \sigma)
  \sim q'$.  Consequently, the errors are exactly $\bigcup_{q \in B}
  \{ u \sigma v \mid u \in L(M, q), v \in L(\delta(q, \sigma), M)
  \simdiff L(q', M) \} \subseteq W_{B, \sigma}$, which yields the
  $\sum_{q \in B} w_q \cdot E_{\delta(q, \sigma), q'}$~errors
  of~$W_{B, \sigma}$ because the decomposition is unique.
\end{proof}

Recall that $w_q$~and~$E_{q,p}$ have been pre-computed already.  Next,
we discuss the full merging algorithm (see
Algorithm~\ref{alg:OptMerge}).  The initial state is handled in
\mbox{lines~1--2}.  In lines~5--7 we first handle the already
discussed decision for the finality of blocks~$B$ of preamble states
and perform the best merge into state~$q$.  In lines~8--11 we
determine the best target state for all transitions from a preamble to
a kernel state.  The smallest error count is added to the global error
count in line~10 and the corresponding designated kernel state is
selected as the new target of the transition in line~11.  This makes
all preamble states that are almost-equivalent to this kernel state
unreachable, so they can be removed.  On our example
\textsc{dfa}~$M_{\text{ex}}$ of Fig.~\ref{fig:Ex}, we have that
$\delta(C, a) = G$ is a transition from the block $\{C, D\} \in
P_\sim$ to a kernel state.  Consequently, we compare $\sum_{q \in \{C,
  D\}} w_q \cdot E_{\delta(q, a), q'}$ for all kernel states~$q' \in K_G$:
\[ \sum_{q \in \{C, D\}} w_q \cdot E_{\delta(q, a), I} = 2 \cdot 1 + 1
\cdot 1 = 3 \quad \text{and} \quad \sum_{q \in \{C, D\}} w_q \cdot
E_{\delta(q, a), J} = 2 \cdot 4 + 1 \cdot 4 = 12 \enspace. \]  

\begin{algorithm}[t]
  \begin{algorithmic}[2]
    \REQUIRE a minimal \textsc{dfa}~$M = (Q, \Sigma, q_0, \delta, F)$
    and its almost-equivalent states~$\sim$
    \ENSURE error count~$e$; initially~$0$
    \smallskip
    \IF{$K_{q_0} \neq \emptyset$}
      \RETURN $\langle (Q, \Sigma, \argmin_{q \in K_{q_0}} E_{q_0, q},
      \delta, F), \min_{q \in K_{q_0}} E_{q_0, q} \rangle$ \\
        \smallskip
    \ENDIF

    \STATE $N \gets M$ where $N = (P, \Sigma, p_0, \mu, G)$
      \COMMENT{initialize output \textsc{dfa}}
    \FORALL{$B \in P_\sim$}
      \STATE $q \gets \textsc{CompFinality}(M, B)$
        \COMMENT{determine finality of merged state}
      \FORALL{$p \in B$}
        \STATE $N \gets \merge_N(p \to q)$
          \COMMENT{perform the merges} \smallskip
      \ENDFOR
      \FORALL{$\sigma \in \Sigma$}
        \IF{$K = K_{\delta(q, \sigma)} \neq \emptyset$}
          \STATE $\displaystyle e \gets e + \min_{q \in K} \Bigl(
          \sum_{p \in B} w_p \cdot E_{\delta(p, \sigma), q} \Bigr)$
            \COMMENT{add best error count}
          \STATE $\displaystyle \mu(q, \sigma) \gets \argmin_{q
            \in K} \Bigl( \sum_{p \in B} w_p \cdot E_{\delta(p,
            \sigma), q} \Bigr)$
            \COMMENT{update follow state}
        \ENDIF
      \ENDFOR
    \ENDFOR
    \RETURN $(N, e)$
  \end{algorithmic}
  \caption{\protect{\textsc{OptMerge}: Optimal merging of
      almost-equivalent states.}}
  \label{alg:OptMerge}
\end{algorithm}

\begin{theorem}
  \label{thm:Correct}
  Algorithm~\ref{alg:OptMerge} runs in time~$O(mn)$ and returns a
  hyper-optimal dfa for~$M$.  In addition, the number of committed
  errors is returned.
\end{theorem}

\begin{proof}
  The time complexity is easy to check, so we leave it as an exercise.
  Since the choices (finality, transition target, initial state) are
  independent by Lemma~\ref{lm:Part}, all hyper-minimal \textsc{dfa}s
  for~$M$ are considered in Algorithm~\ref{alg:OptMerge} by
  Theorem~\ref{thm:Struc}.  Consequently, we can always select the
  local optimum for each choice (using Lemmata \ref{lm:Initial},
  \ref{lm:Finality}, and \ref{lm:Trans}) to obtain a global optimum,
  which proves that the returned number is the minimal number of
  errors among all hyper-minimal \textsc{dfa}s.  Mind that the number
  of errors would be infinite for a hyper-minimal \textsc{dfa} that is
  not almost-equivalent to~$M$.  Moreover, it is obviously the number
  of errors committed by the returned \textsc{dfa}, which proves that
  the returned \textsc{dfa} is hyper-optimal for~$M$.
\end{proof}

\begin{corollary}[{\protect{of Theorem~\ref{thm:Correct}}}]
  \label{cor:Main}
  For every \textsc{dfa}~$M$ we can obtain a hyper-optimal
  \textsc{dfa} for~$M$ in time~$O(mn)$.
\end{corollary}


\section{Empirical results}
\label{sec:Empirical}
In order to evaluate the algorithm, we compare it to another
hyper-minimization algorithm \cite{holmal10} that does not aim for low
error profile.  Since the algorithm of~\cite{holmal10} is
(``don't-care'') non-deterministic (in the selection of merge
targets), we implemented a simple stack discipline, which always pops
the first element.  For a varying set of parameters, 100~random
\textsc{dfa}s have been generated and run through both algorithms. The
number of saved states as well as the number of errors are reported.
First we explain how the test \textsc{dfa}s were generated, describe
the experimental setup, and then present and discuss the results.

We use an algorithm based on the original algorithm in
\textsc{Hanneforth}'s FSM<2.0> library~\cite{han09}, which generates
random non-deterministic finite automata. This model is closely
related to \textsc{Karp}'s model of random directed graphs (see
Chapter~2 of~\cite{bol01} or~\cite{tabvar05} for a discussion of
different models).  The only difference is the introduction of an
additional parameter: the \emph{cyclicity}~$a$.  The complete set of
parameters is as follows:
\begin{itemize}
\item[$\abs Q$] This integer limits the number of states in
  the non-deterministic automaton.
\item[$\abs \Sigma$] This integer coincides with the number of
  alphabet symbols.
\item[$d_\delta$] Uniform probability determining whether a given
  transition $p \stackrel\sigma\to q$ exists; we call
  $d_\delta \cdot \abs Q$ the transition density. 
\item[$d_F$] Uniform probability for a given state to be final.
\item[$a$] This real-valued parameter~$0 \leq a \leq 1$ controls the
  cyclicity by constraining ``backward-pointing'' transitions.  In
  particular, if~$a = 0$, then the automaton will be acyclic, and
  if~$a = 1$, then all transitions are equally probable.
\end{itemize}
A non-deterministic automaton $M$ is generated in the following way:
(i)~The set of states is $Q = \{0, 1, 2, \dotsc , \abs Q - 1\}$ with
initial state~$0$.  (ii)~A state~$q \in Q$ is final if and only if
$f_q < d_F$, where $0 \leq f_q \leq 1$ is a random
value. (iii)~Finally, for every $(q, a, p) \in Q \times \Sigma \times
Q$, we generate a random number $0 \leq f_{(q, a, p)} \leq 1$.  The
transition $q \stackrel\sigma\to p$ is present in~$M$ if and only if
\[ f_{(q, a, p)} <
\begin{cases}
  d_\delta & \text{if } p > q \\
  a \cdot d_\delta & \text{otherwise.}
\end{cases} \] The latter case corresponds to ``backward-pointing''
transitions and creates cycles.

For each set of parameters, we have generated 100~\textsc{dfa}s.
These \textsc{dfa}s were obtained by determinizing and minimizing the
randomly generated non-deterministic test automata.  All \textsc{dfa}s
have then been hyper-minimized, and the optimal hyper-minimal
\textsc{dfa}s have been compared to the ones resulting from na\"\i ve
hyper-minimization.\footnote{The complete C${}^{++}$~source code will
  be made available, and the FSM<2.0> library is available at
  \url{http://tagh.de/tom/?p=1737}.}  The obtained results are shown in
Figs.\ \ref{fig:Graph1}~and~\ref{fig:Graph2}.

\begin{figure}[t]
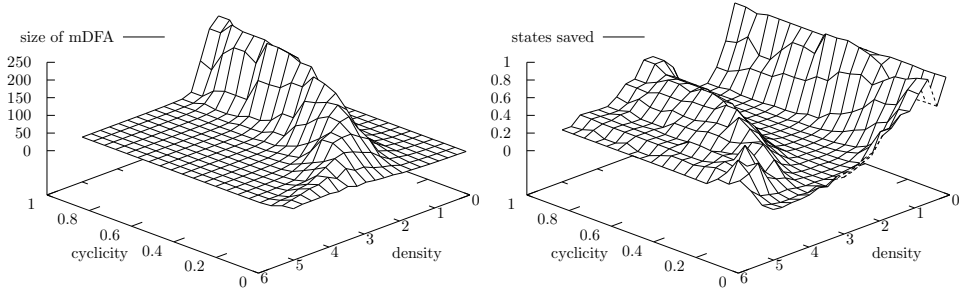

  \centering
  \includegraphics[scale=0.6]{figure/plot0.mps} \hfill
  \includegraphics[scale=0.6]{figure/plot1.mps} 
  \caption{Hyper-minimization performance for non-deterministic
    automata with 30 states, $|\Sigma| = 2$ and $0.3 \leq d_F \leq 0.7$.
    ``Density'' refers to $d_\delta \cdot \abs Q$.  Left: Average size
    of the minimal \textsc{dfa}.  Right: Ratio of states saved by
    hyper-minimization.  Values range over the full scale; i.e., 
    they approach~$0$ outside the ridge and inside the valley.}
  \label{fig:Graph1}
\end{figure}

Figure~\ref{fig:Graph1} shows the size of the minimal \textsc{dfa}s
and the potential of saving states by hyper-minimization.  The left
graph in Fig.~\ref{fig:Graph1} shows a ridge, which corresponds to
cases in which \textsc{dfa} minimization is hard and results in a
large minimal \textsc{dfa}~\cite{tabvar05}.  It is located around a
transition density of~$d_\delta \cdot \abs Q = 1.25$ for a cyclicity
of~$1$, and it moves to higher densities for less cyclic automata.
Essentially, the same ridge was observed by~\cite{tabvar05} (for the
case~$a = 1$).  The right graph in Fig.~\ref{fig:Graph1} shows that
these hard instances for \textsc{dfa} minimization are also hard for
hyper-minimization in the sense that only very few states can be
saved.  However, for the remaining instances a considerable reduction
in the number of states is achievable by hyper-minimization.

\begin{figure}[t]
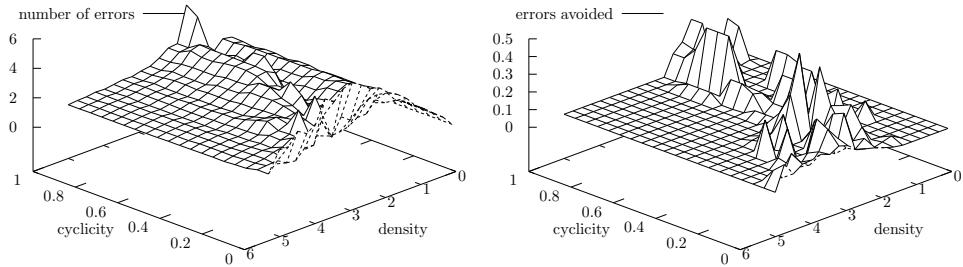

  \centering
  \includegraphics[scale=0.6]{figure/plot2.mps} \hfill
  \includegraphics[scale=0.6]{figure/plot3.mps} 
  \caption{Hyper-optimization performance. Left: Absolute number of
    errors in na\"\i ve hyper-minimal \textsc{dfa}s. Right: Ratio of
    errors avoided by hyper-optimization.} 
  \label{fig:Graph2}
\end{figure}

If we focus on the contribution of this paper, then we find that the
number of errors can be considerably reduced. Figure~\ref{fig:Graph2}
shows the absolute number of errors for hyper-minimal \textsc{dfa}s
(left graph) and the ratio of errors avoided by the hyper-optimal
automaton (right graph).  The absolute number of errors for the hard
instances, which can only be reduced a little, is higher than for the
easy instances.  However, the hyper-optimal \textsc{dfa}s avoid a
higher ratio of errors for the hard instances, which dramatically
reduces the number of committed errors paid for the
small reduction.


\bibliography{extra}
\bibliographystyle{ijfcs}

\end{document}